\documentclass[12pt]{article}
\usepackage{algorithm,algorithmic,amsmath,amssymb,bm,caption,subcaption,graphicx,placeins,color,natbib,multirow,amsthm,fullpage,pdflscape,ragged2e}
\usepackage{pdflscape}

\usepackage{mathtools}

\newtheorem{theorem}{Theorem}
\newtheorem{remark}{Remark}
\newtheorem{lemma}{Lemma}

\newtheorem{proposition}{Proposition}

\title{Control Functionals for Monte Carlo Integration}
\author{Chris J. Oates$^{1,2,}$\footnote{{\it Address for correspondence:} School of Mathematical and Physical Sciences, University of Technology Sydney, NSW 2007, Australia. E-mail: christopher.oates@uts.edu.au}, Mark Girolami$^{3,4}$ and Nicolas Chopin$^5$ \\
$^1$University of Technology Sydney, Australia \\
$^2$Australian Research Council Centre for Excellence in Mathematical \\ and Statistical Frontiers \\
$^3$University of Warwick, Coventry, UK \\
$^4$Alan Turing Institute \\
$^5$CREST-LS and ENSAE, Paris, France}

\begin{document}
\maketitle

\begin{abstract}
A non-parametric extension of control variates is presented.
These leverage gradient information on the sampling density to achieve substantial variance reduction.
It is not required that the sampling density be normalised.
The novel contribution of this work is based on two important insights; (i) a trade-off between random sampling and deterministic approximation and (ii) a new gradient-based function space derived from Stein's identity. 
Unlike classical control variates, our estimators achieve super-root-$n$ convergence, often requiring orders of magnitude fewer simulations to achieve a fixed level of precision.
Theoretical and empirical results are presented, the latter focusing on integration problems arising in hierarchical models and models based on non-linear ordinary differential equations.
\end{abstract}

\section{Introduction} \label{intro}

Statistical methods are increasingly being employed to analyse complex models of physical phenomena \citep[e.g. in climate forecasting or simulations of molecular dynamics;][]{Slingo,Angelikopoulos}.
Analytic intractability of complex models has inspired the development of sophisticated Monte Carlo methodologies to facilitate computation \citep{Robert}.
In their most basic form, Monte Carlo estimators converge as the reciprocal of root-$n$ where $n$ is the number of random samples.
For complex models it may only be feasible to obtain a limited number of samples \citep[e.g. a recent Met Office model for future climate simulations required the order of $10^6$ core-hours per simulation;][]{Mizielinski}.
In these situations, root-$n$ convergence is too slow and leads in practice to high-variance estimation.
Our contribution is motivated by resolving this issue and provides novel methodology that is both formal and general.

The focus of this paper is the estimation of an expectation $\mu(f) = \int f(\bm{x}) \pi(\bm{x}) \mathrm{d}\bm{x}$, where $f$ is a test function of interest and $\pi$ is a probability density associated with a random variable $\bm{X}$.
Provided that $f(\bm{X})$ has variance $\sigma^2(f) < \infty$, the arithmetic mean estimator
\begin{eqnarray*}
\frac{1}{n}\sum_{i=1}^n f(\bm{x}_i), 
\end{eqnarray*}
based on $n$ independent and identically distributed (IID) samples $\{\bm{x}_i\}_{i=1}^n$ of the random variable, satisfies the central limit theorem and converges to $\mu(f)$ at the rate $O_P(n^{-1/2})$, or simply at ``root-$n$''.
When working with complex models, root-$n$ convergence can be problematic, as highlighted in e.g. \cite{Ba}.
A model is considered complex when either (i) $\bm{X}$ is expensive to simulate, or (ii) $f$ is expensive to evaluate, in each case relative to the required estimator precision.
Both situations are prevalent in scientific and engineering applications \citep[e.g.][]{Kohlhoff,Higdon}.
This paper introduces a class of estimators that converge more quickly than root-$n$.
The significance of our contribution is made clear in the comparative overview below.

Generic approaches to reduction of variance are well-known in both statistics and numerical analysis.
These include (i) importance sampling and its extensions \citep{Cornuet,Li2}, (ii) stratified sampling and related techniques \citep{Rubinstein2}, (iii) antithetic variables \citep{Green2} and more generally (randomised) quasi-Monte Carlo \citep[QMC/RQMC;][]{Dick}, (iv) Rao-Blackwellisation \citep{Robert,Douc,Ghosh,Olsson},  (v) Riemann sums \citep{Philippe}, (vi) control variates \citep{Glasserman,Mira,Li}, (vii) multi-level Monte Carlo and related techniques \citep[e.g.][]{Heinrich2,Giles,Giles2}, (viii) Bayesian Monte Carlo \citep[BMC;][]{OHagan,Rasmussen,Briol}, and (ix) a plethora of sophisticated Markov chain Monte Carlo sampling schemes \citep[MCMC;][]{Latuszynski2}.
Classical introductions to many of the above techniques include \citet[][Chap. 4]{Robert} and \citet[][Chap. 5]{Rubinstein2}.

Motivated by contemporary statistical applications, we state four {\it desiderata} for a variance reduction technique:
(I) {\it Unbiased estimation:} Monte Carlo (MC) methods based on IID samples produce unbiased estimators, whilst techniques such as MCMC generally produce biased estimators.
(II) {\it Compatibility with an un-normalised density $\pi$:} An ``un-normalised'' density is known only up to proportionality so that, for example, MCMC techniques are required for sampling.
(III) {\it Super-root-$n$ convergence} (for sufficiently regular $f$): The convergence rates of (R)QMC are well studied and can be super-root-$n$. 
Riemann sums can also achieve super-root-$n$ rates and \cite{Briol2} showed the same holds for BMC.
(IV) {\it Post-hoc schemes:} Rao-Blackwellisation, Riemann sums, BMC and control variates can all be conceived as {\it post-hoc} schemes; i.e. schemes that can be applied retrospectively after samples have been obtained. 
In contrast, the remaining methods require modification to computer code for the sampling process itself.
The former are appealing from both a theoretical and a practical perspective since they separate the challenge of sampling from the challenge of variance reduction.

\begin{table}

\centering
\resizebox{\textwidth}{!}{
\begin{tabular}{p{7cm}ccccc}
\hline
Estimation Method & Unbiased & Un-normalised $\pi$ & Super-root-$n$ & Post-hoc \\ \hline
MC(/MCMC) + Arithmetic Mean  & \checkmark(/$\times$) & $\times$(/\checkmark) & $\times$ & $\times$ \\
MC + Importance Sampling   & \checkmark(/$\times$) & $\times$(/\checkmark) & $\times$ & $\times$ \\ 
MC + Antithetic Variables   & \checkmark & $\times$ & $\times$ & $\times$ \\
MC(/MCMC) + Stratified Sampling  & \checkmark(/$\times$) & $\times$(/\checkmark)  & $\times$ & $\times$ \\
Quasi-MC (QMC)   & $\times$ & $\times$ & \checkmark & $\times$  \\ 
Randomised QMC (RQMC)   & \checkmark & $\times$ & \checkmark & $\times$ \\ 
MC(/MCMC) + Rao-Blackwellisation  & \checkmark(/$\times$) & $\times$(/\checkmark) & $\times$ & \checkmark \\
MC(/MCMC) + Control Variates  & \checkmark(/$\times$)  &  $\times$(/\checkmark) & $\times$ & \checkmark \\
MC(/MCMC) + Riemann Sums   & $\times$ & $\times$(/\checkmark) & \checkmark & \checkmark \\ 
Bayesian MC (BMC) & $\times$ & $\times$ & \checkmark & \checkmark \\ \hline
MC(/MCMC) + Control Functionals   & \checkmark($/\times$) & $\times$(/\checkmark) & \checkmark & \checkmark \\ \hline
\end{tabular}}
\caption{A comparison of estimation methods for integrals. 
[``Unbiased'' = the estimator is unbiased for $\mu(f)$.
``Un-normalised $\pi$'' = the estimator can handle sampling densities that are only available up to proportionality.
``Super-root-$n$'' = the estimator converges faster than root-$n$.
``Post-hoc'' = the estimator places no restriction on how the samples $\bm{x}_i$ are generated, i.e. requires no modification to computer code for sampling.
Estimator properties may change in order to handle un-normalised densities $\pi$; these are shown in parentheses.]
}
\label{glass}

\end{table}

Table \ref{glass} summarises existing techniques in relation to these {\it desiderata}; note that no technique fulfils all four criteria.
In contrast, the method proposed here, called ``control functionals'', is able to satisfy all four {\it desiderata}.
Control functionals appear to be similar, in this sense, to Riemann sums i.e. they are a super-root-$n$, {\it post-hoc} approach that applies to un-normalised sampling densities.
However, Riemann sums are rarely used in practice due to (i) the fact that estimators are biased at finite sample sizes, and (ii) there is a prohibitive increase in methodological complexity for multi-dimensional state spaces.
Control functionals do not posses either of these drawbacks.

The control functional method that we develop below can be intuitively considered as a non-parametric development of control variates.
In control variate schemes one seeks a basis $\{s_i\}_{i=1}^m$, $m \in \mathbb{N}$, that have expectation $\mu(s_i) = 0$.
Then a surrogate function $\tilde{f} = f - a_1s_1 - \dots - a_ms_m$ is constructed such that $\mu(\tilde{f}) = \mu(f)$ and, for suitably chosen $a_1,\dots,a_m \in \mathbb{R}$, a variance reduction $\sigma^2(\tilde{f}) < \sigma^2(f)$ is obtained \citep[see e.g.][]{Rubinstein}.
The statistics $s_i$ are known as control variates and the variance $\sigma^2(\tilde{f})$ can be reduced to zero if and only if there is perfect canonical correlation between $f$ and the basis $\{s_i\}_{i=1}^m$.
For estimation based on Markov chains, control variates for the discrete state space case were provided by \cite{Andradottir}.
For continuous state spaces, statistics relating to the chain can be used as control variates \citep{Hammer,Dellaportas,Li}.
Alternatively control variates can be constructed based on gradient information \citep{Assaraf,Mira}.

The control variates described above are solving a misspecified regression problem, since in general $f$ will not be a linear combination of the $s_i$ basis functions.
As such they achieve at most a constant factor reduction in estimator variance.
Intuitively, one would like to increase the number $m$ of basis functions to increase in line with the number $n$.
\cite{Mijatovic} explored this approach within the Metropolis-Hastings method.
However, their solution requires the user to partition of the state space, which limits its wider appeal.
This paper introduces a powerful new perspective on variance reduction that fully resolves these issues, satisfying all the {\it desiderata} described above.
To realise our method we developed a gradient-based function space that leads to closed-form estimators whose convergence can be guaranteed.
The functional analysis perspective works ``out of the box'', without requiring the user to partition the state space.
Extensive empirical support is provided in favour of the proposed method, including applications to hierarchical models and models based on non-linear differential equations.
In each case state-of-the-art estimation is achieved.

All results can be reproduced using \verb+MATLAB R2015a+ code that is available to download from \verb+http://warwick.ac.uk/control_functionals+.

\section{Methodology} \label{methods}

\subsection{Set-up and notation}

Consider a random vector $\bm{X}$ taking values in an open set $\Omega \subseteq \mathbb{R}^d$.
Assume $\bm{X}$ admits a positive density on $\Omega$ with respect to $d$-dimensional Lebesgue measure,  written $\pi(\bm{x})>0$.
For bounded $\Omega$ with boundary $\partial \Omega$, we assume $\partial\Omega$ is piecewise smooth (i.e. infinitely differentiable).
Write $\mathcal{L}^2(\pi)$ for the space of measurable functions $g:\Omega \rightarrow \mathbb{R}$ for which $\int_{\Omega} g(\bm{x})^2 \pi(\bm{x}) \mathrm{d}\bm{x}$ is finite. 
Write $C^k(\Omega,\mathbb{R}^j)$ for the space of (measurable) functions from $\Omega$ to $\mathbb{R}^j$ with continuous partial derivatives up to order $k$.
Consider a test function $f : \Omega \rightarrow \mathbb{R}$ of interest, assume $f \in \mathcal{L}^2(\pi)$ and write $\mu(f) := \int_{\Omega} f(\bm{x}) \pi(\bm{x}) \mathrm{d}\bm{x}$, $\sigma^2(f) := \int_{\Omega} (f(\bm{x}) - \mu(f))^2 \pi(\bm{x}) \mathrm{d}\bm{x}$.

Denote by $\mathcal{D} = \{\bm{x}_i\}_{i=1}^n$ a collection of states $\bm{x}_i \in \Omega$.
At each state $\bm{x}_i$ the corresponding function values $f(\bm{x}_i)$ and gradients $\nabla_{\bm{x}} \log \pi(\bm{x}_i)$ are assumed to have been pre-computed and cached. 
The method that we develop does not then require any further recourse to the statistical model $\pi$, nor any further evaluations of the function $f$, and is in this sense a widely-applicable {\it post-hoc} scheme.

\subsection{From control variates to control functionals}

\subsubsection{Deterministic approximation}

Our starting point is establish a trade-off between random sampling and deterministic approximation, as suggested on several separate occasions by authors including \cite{Bakhvalov,Heinrich2,Speight,Giles}.

Consider a dichotomy of available states $\mathcal{D} = \{\bm{x}_i\}_{i=1}^n$ into two disjoint subsets $\mathcal{D}_0 = \{\bm{x}_i\}_{i=1}^m$ and $\mathcal{D}_1 = \{\bm{x}_i\}_{i=m+1}^n$, where $1 \leq m < n$.
Although $m$, $n$ are fixed, we will be interested in the asymptotic regime where $m = O(n^\gamma)$ for some $\gamma \in [0,1]$.
Consider surrogate functions of the form
\begin{eqnarray*}
f_{\mathcal{D}_0}(\bm{x}) := f(\bm{x}) - s_{f,\mathcal{D}_0}(\bm{x}) + \mu(s_{f,\mathcal{D}_0}), \label{tradeoff}
\end{eqnarray*}
where $s_{f,\mathcal{D}_0} \in \mathcal{L}^2(\pi)$ is an approximation to $f$, based on $\mathcal{D}_0$, whose expectation $\mu(s_{f,\mathcal{D}_0})$ is analytically tractable.
By construction $f_{\mathcal{D}_0} \in \mathcal{L}^2(\pi)$, $\mu(f_{\mathcal{D}_0}) = \mu(f)$ and $\sigma^2(f_{\mathcal{D}_0}) = \sigma^2(f - s_{f,\mathcal{D}_0})$.
We study estimators of the form
\begin{eqnarray*}
\hat{\mu}(\mathcal{D}_0,\mathcal{D}_1;f) := \left\{ \begin{array}{ll} \frac{1}{n-m} \sum_{i=m+1}^{n} f_{\mathcal{D}_0}(\bm{x}_{i}) & \text{for } m < n \\ \mu(s_{f,\mathcal{D}_0}) & \text{for } m = n. \end{array} \right.
\end{eqnarray*}
For theoretical purposes the second subset $\mathcal{D}_1$ is assumed to be an IID sample from $\pi$, statistically independent from $\mathcal{D}_0$.
Then, for $m < n$, we have unbiasedness, i.e. $\mathbb{E}_{\mathcal{D}_1}[\hat{\mu}(\mathcal{D}_0,\mathcal{D}_1;f)] = \mu(f)$, where the expectation here is with respect to the sampling distribution $\pi$ of the $n-m$ random variables that constitute $\mathcal{D}_1$, and is conditional on $\mathcal{D}_0$.
The corresponding estimator variance, conditional on $\mathcal{D}_0$, is $\mathbb{V}_{\mathcal{D}_1}[\hat{\mu}(\mathcal{D}_0,\mathcal{D}_1;f)] = (n-m)^{-1} \sigma^2(f - s_{f,\mathcal{D}_0})$.
This formulation encompasses control variates as the special case where $s_{f,\mathcal{D}_0}$ is constrained to a finite-dimensional space.

The insight required to go beyond control variates and achieve super-root-$n$ convergence is that we can use an infinite-dimensional space to construct an increasingly accurate approximations $s_{f,\mathcal{D}_0}$ as $m \rightarrow \infty$.
We allow for the possibility that the first subset $\mathcal{D}_0$ are also random and write $\mathbb{E}_{\mathcal{D}_0}$ to denote an expectation with respect to the (marginal) distribution of these $m$ random variables.

\begin{proposition} \label{theo1}
Assume $m = O(n^\gamma)$ for some $\gamma \in [0,1]$ and that the expected functional approximation error (EFAE) vanishes as
\begin{eqnarray*}
\mathbb{E}_{\mathcal{D}_0}[\sigma^2(f - s_{f,\mathcal{D}_0})] = O(m^{-\delta}) \label{EFAE eqn}
\end{eqnarray*}
for some $\delta \geq 0$.
Then $\mathbb{E}_{\mathcal{D}_0} \mathbb{E}_{\mathcal{D}_1}[(\hat{\mu}(\mathcal{D}_0,\mathcal{D}_1;f) - \mu(f))^2] = O(n^{-1-\gamma\delta})$.
\end{proposition}
\noindent All proofs are reserved for Appendix \ref{aproofs}.
\begin{remark}
Taking $\gamma = 1$ optimises the rate in Prop. \ref{theo1} and we therefore assume in the sequel that $m/n \rightarrow r \in (0,1)$.
\end{remark}

\subsubsection{Control variates based on Stein's identity}

To construct approximations $s_{f,\mathcal{D}_0}$ whose integrals $\mu(s_{f,\mathcal{D}_0})$ are analytically tractable, we make the assumption 
\begin{enumerate}
\item[(A1)] The density $\pi$ belongs to $C^{1}(\Omega,\mathbb{R})$.
\end{enumerate}
Denote the gradient function by $\bm{u}(\bm{x}) := \nabla_{\bm{x}} \log \pi(\bm{x})$ where $\nabla_{\bm{x}} : = [\partial/\partial x_1,\dots,\partial/\partial x_d]^T$, well-defined by (A1).
We study approximations of the form
\begin{eqnarray}
s_{f,\mathcal{D}_0}(\bm{x}) & := & c + \psi(\bm{x})  \nonumber \\
\psi(\bm{x}) & := & \nabla_{\bm{x}} \cdot \bm{\phi}(\bm{x}) + \bm{\phi}(\bm{x}) \cdot \bm{u}(\bm{x}) \label{cfdef2}
\end{eqnarray}
where $c \in \mathbb{R}$ is a constant and $\bm{\phi} \in C^1(\Omega,\mathbb{R}^d)$.
Eqn. \ref{cfdef2} appears in Stein's classical test for approximate normality \citep{Stein} and related to (but simpler than) control variates proposed by \cite{Assaraf,Mira}.
We make the following assumption \citep[c.f. e.g. Eqn. 9 of][]{Mira}:
\begin{enumerate}
\item[(A2)] Let $\bm{n}(\bm{x})$ be the unit normal to the boundary $\partial\Omega$ of the state space $\Omega$. Then
\begin{eqnarray*}
\oint_{\partial\Omega} \pi(\bm{x})\bm{\phi} (\bm{x}) \cdot \bm{n}(\bm{x}) S(\mathrm{d}\bm{x}) = 0.
\end{eqnarray*}
\end{enumerate}
(The notation $\oint_{\partial\Omega}$ denotes a surface integral over $\partial\Omega$ and $S(\mathrm{d}\bm{x})$ denotes the surface element at $\bm{x} \in \partial\Omega$.)
Stein's identity implies that this class of approximations has integrals that are analytically tractable:
\begin{proposition} \label{div prop}
Assume (A1,2). Then $\mu(\psi) = 0$ and so $\mu(s_{f,\mathcal{D}_0}) = c$.
\end{proposition}
\noindent When $\Omega$ is unbounded, all surface integrals are interpreted as tail conditions.
i.e. (A2) should be replaced with $\oint_{\Gamma_r \cap \Omega} \pi(\bm{x}) \bm{\phi}(\bm{x}) \cdot \bm{n}(\bm{x}) S(\mathrm{d}\bm{x}) \rightarrow 0$, where $\Gamma_r \subset \mathbb{R}^d$ is the sphere of radius $r$ centred at the origin and $\bm{n}(\bm{x})$ is the unit normal to the surface of $\Gamma_r$.

The statistic $\psi$ is recognised as a control variate.
These control variates were explored in the case where $\bm{\phi}$ is a (gradient of a low-degree) polynomial by \cite{Assaraf}, \cite{Assaraf2} and \cite{Mira}.
This paper takes the innovative step of setting $\bm{\phi}$ within a function space to enable fully non-parametric approximation.
The functional approximation perspective differs fundamentally from the control variate approach, in which the estimation problem is formally mis-specified (i.e. $\bm{\phi}$ is restricted to a low dimensional parametric family that does not contain the ``true'' function).
We emphasise this key conceptual distinction by referring to $\psi$ as a control \emph{functional} (CF; reflecting the use of terminology from functional analysis).

\subsection{Theory}

This section establishes $\psi$ as belonging to a Hilbert space $\mathcal{H}_0 \subset \mathcal{L}^2(\pi)$. 
This allows us to formulate and solve a functional approximation problem that targets the EFAE.

\subsubsection{A Hilbert space of control functionals}

Specification of $\psi$ is equivalent to specification of $\bm{\phi}$.
We decide to restrict each component function $\phi_i : \Omega \rightarrow \mathbb{R}$ to a Hilbert space $\mathcal{H} \subset \mathcal{L}^2(\pi) \cap C^{1}(\Omega,\mathbb{R})$ with inner product $\langle \cdot,\cdot \rangle_{\mathcal{H}} : \mathcal{H} \times \mathcal{H} \rightarrow \mathbb{R}$.
Moreover we insist that $\mathcal{H}$ is a reproducing kernel Hilbert space.
This implies that there exists a symmetric positive definite function $k : \Omega \times \Omega \rightarrow \mathbb{R}$ such that (i) for all $\bm{x} \in \Omega$ we have $k(\cdot,\bm{x}) \in \mathcal{H}$ and (ii) for all $\bm{x} \in \Omega$ and $h \in \mathcal{H}$ we have $h(\bm{x}) = \langle h , k(\cdot,\bm{x}) \rangle_{\mathcal{H}}$ \citep[][Def. 1, p7, and Def. 2, p10]{Berlinet}.
The vector-valued function $\bm{\phi} : \Omega \rightarrow \mathbb{R}^d$ is defined in the Cartesian product space $\mathcal{H}^d := \mathcal{H} \times \dots \times \mathcal{H}$, itself a Hilbert space with the inner product $\langle \bm{\phi}, \bm{\phi}' \rangle_{\mathcal{H}^d} = \sum_{i=1}^d \langle \phi_i, \phi_i' \rangle_{\mathcal{H}}$.

We make an assumption on $k$ that will be enforced by construction:
\begin{enumerate}
\item[(A3)] The kernel $k$ belongs to $C^{2}(\Omega \times \Omega,\mathbb{R})$.
\end{enumerate}
Now we can analyse the class of CFs induced by $k$:
\begin{theorem} \label{conjugate}
Assume $\bm{\phi} \in \mathcal{H}^d$ and (A1,3).
Then $\psi$ belongs to $\mathcal{H}_0$, the reproducing kernel Hilbert space with kernel $k_0(\bm{x},\bm{x}') := \nabla_{\bm{x}} \cdot \nabla_{\bm{x}'} k(\bm{x},\bm{x}') + \bm{u}(\bm{x}) \cdot \nabla_{\bm{x}'} k(\bm{x},\bm{x}')  + \bm{u}(\bm{x}') \cdot \nabla_{\bm{x}} k(\bm{x},\bm{x}') + \bm{u}(\bm{x}) \cdot \bm{u}(\bm{x}') k(\bm{x},\bm{x}')$.
\end{theorem}

To gain some intuition for $\mathcal{H}_0$ we strengthen (A2) as follows:
\begin{enumerate}
\item[(A2')] For $\pi$-almost all $\bm{x} \in \Omega$ the kernel $k$ satisfies
\begin{eqnarray*}
\oint_{\partial\Omega} k(\bm{x},\bm{x}') \pi(\bm{x}') \bm{n}(\bm{x}') S(\mathrm{d}\bm{x}') = \bm{0}
\end{eqnarray*}
and 
\begin{eqnarray*}
\oint_{\partial\Omega} \nabla_{\bm{x}} k(\bm{x},\bm{x}') \pi(\bm{x}') \cdot \bm{n}(\bm{x}') S(\mathrm{d}\bm{x}') = 0.
\end{eqnarray*}
\end{enumerate}
\noindent While (A2') must be verified on a case-by-case basis, it can in principle always be enforced with a suitable choice of $k$. 

\begin{lemma} \label{mean ele lem}
Under (A1,2',3), the gradient-based kernel $k_0$ satisfies
\begin{eqnarray*}
\int_{\Omega} k_0(\bm{x},\bm{x}') \pi(\bm{x}') \mathrm{d}\bm{x}' = 0 \label{mean element}
\end{eqnarray*}
for $\pi$-almost all $\bm{x} \in \Omega$.
\end{lemma}
\noindent Lemma \ref{mean ele lem} generalises Eqn. 1 of \cite{Mira} and implies that $\mathcal{H}_0$ consists of only valid CFs, i.e. $\psi \in \mathcal{H}_0 \implies \mu(\psi) = 0$.
These ideas are illustrated in Fig. \ref{space}.

\begin{figure}[t!]
\centering
\includegraphics[width = \textwidth,clip,trim = 0.2cm 0cm 0cm 0cm]{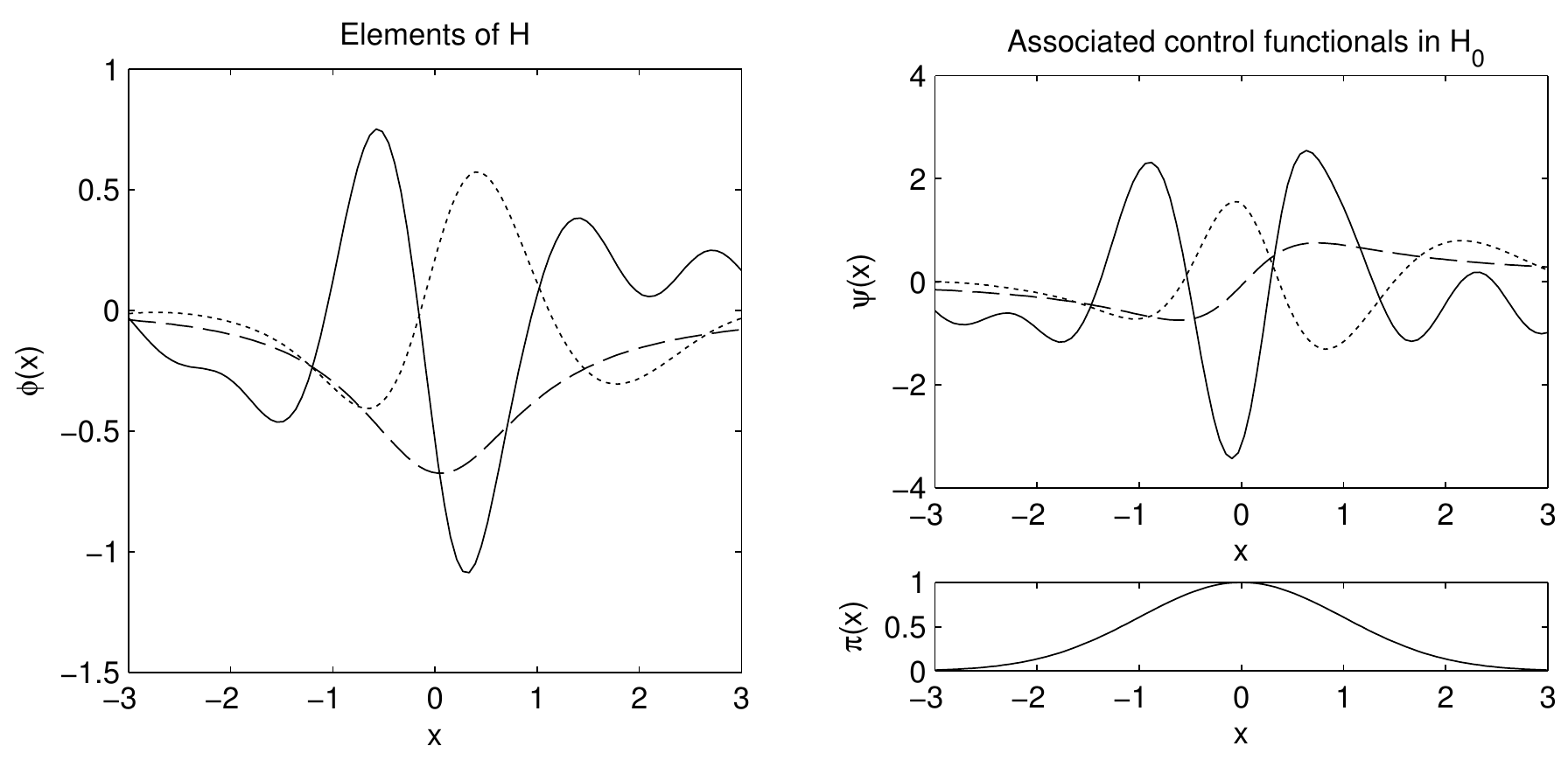}
\caption{Constructing control functionals (in dimension $d=1$):
Representative elements $\phi$ from the reproducing kernel Hilbert space $\mathcal{H}$ (left panel) are plotted, along with their associated control functionals $\psi = \nabla \phi + \phi \nabla \log \pi$ in $\mathcal{H}_0$ (right, top panel).
Each $\phi$ is unconstrained in expectation, but the corresponding control functional $\psi$ is automatically constrained to have expectation zero with respect to the (possibly un-normalised) probability density $\pi$ (right, bottom panel).}
\label{space}
\end{figure}

\begin{enumerate}
\item[(A4)] The gradient-based kernel $k_0$ satisfies
\begin{eqnarray*}
\int_{\Omega} k_0(\bm{x},\bm{x}) \pi(\bm{x}) \mathrm{d}\bm{x} < \infty.
\end{eqnarray*}
\end{enumerate}

\begin{lemma} \label{finite variance}
Under (A1,2',3,4) we have $\mathcal{H}_0 \subset \mathcal{L}^2(\pi)$.
\end{lemma}
\begin{remark}
In general (A4) must be verified on a case-by-case basis. (A4) is easily verified for all examples in this paper.
\end{remark}

\subsubsection{Consistent approximation and asymptotics}

Now we establish theoretical results for consistent approximation of $f$ by $s_{f,\mathcal{D}_0}$.
Write $\mathcal{C}$ for the reproducing kernel Hilbert space of constant functions with kernel $k_{\mathcal{C}}(\bm{x},\bm{x}') = 1$ for all $\bm{x},\bm{x}' \in \Omega$.
Denote the norms associated to $\mathcal{C}$ and $\mathcal{H}_0$ respectively by $\|\cdot\|_{\mathcal{C}}$ and $\|\cdot\|_{\mathcal{H}_0}$.
Write $\mathcal{H}_+ = \mathcal{C} + \mathcal{H}_0$ for the set $\{c + \psi : c \in \mathcal{C}, \; \psi \in \mathcal{H}_0\}$.
Equip $\mathcal{H}_+$ with the structure of a vector space, with addition operator $(c + \psi) + (c' + \psi') = (c + c') + (\psi + \psi')$ and multiplication operator $\lambda(c + \psi) = (\lambda c) + (\lambda \psi)$, each well-defined due to uniqueness of the representation $f = c + \psi$, $f' = c' + \psi'$ with $c,c' \in \mathcal{C}$ and $\psi,\psi' \in \mathcal{H}_0$.
In addition, equip $\mathcal{H}_+$ with the norm $\|f\|_{\mathcal{H}_+}^2 := \|c\|_{\mathcal{C}}^2 + \|\psi\|_{\mathcal{H}_0}^2$, again, well-defined by uniqueness of representation.
It can be shown that $\mathcal{H}_+$ is a reproducing kernel Hilbert space with kernel $k_+(\bm{x},\bm{x}') := k_{\mathcal{C}}(\bm{x},\bm{x}') + k_0(\bm{x},\bm{x}')$ \citep[][Thm. 5, p24]{Berlinet}.

For the analysis we assume a basic well-posedness condition:
\begin{enumerate}
\item[(A5)] $f \in \mathcal{H}_+$. i.e. $f = c + \psi$ for some $c \in \mathcal{C}$ and $\psi \in \mathcal{H}_0$.
\end{enumerate}

\begin{remark}
(A5) is equivalent to the existence of a solution $\bm{\phi} \in \mathcal{H}^d$ to the partial differential equation
\begin{eqnarray*}
\nabla_{\bm{x}} \cdot [\pi(\bm{x}) \bm{\phi}(\bm{x})] = [f(\bm{x}) - \mu(f)] \pi(\bm{x}), 
\end{eqnarray*}
called the ``fundamental equation'' in \citet[][Eqn. 5; see also Eqn. 4 in \cite{Mira}]{Assaraf}.
With no initial or boundary conditions to satisfy, it is easy to show that there exist infinitely many solutions to the fundamental equation, so (A5) is automatically satisfied by choosing $k$ such that $\mathcal{H}$ is big enough to contain at least one solution.
\end{remark}

To realise the CF method we consider the regularised least-squares (RLS) functional approximation given by
\begin{eqnarray*}
s_{f,\mathcal{D}_0} := \underset{g \in \mathcal{H}_+}{\arg\min} \left\{ \frac{1}{m} \sum_{j=1}^m (f(\bm{x}_j) - g(\bm{x}_j))^2 + \lambda \|g\|_{\mathcal{H}_+}^2 \right\} \label{RLS estimator}
\end{eqnarray*}
where $\lambda > 0$.
For the special case where $m = n$, the CF estimator can be interpreted as kernel quadrature \citep{Sommariva} and also as empirical interpolation \citep{Kristoffersen}.
The distinguishing feature of CFs from these methods is that the Stein construction is compatible with un-normalised $\pi$.

Below we will establish that the RLS estimate produces vanishing EFAE under a strengthening of (A4):
\begin{enumerate}
\item[(A4')] $\sup_{\bm{x} \in \Omega} k_0(\bm{x},\bm{x}) < \infty$.
\end{enumerate}
\begin{remark} \label{sup condition}
(A4') would follow from (A3) and compactness of the state space $\Omega$, but we do not assume compactness here.
All experiments in this paper have at worst $\bm{u}(\bm{x}) = O(\|\bm{x}\|_2)$, so that (A4') is automatically satisfied, for example, when we choose $k(\bm{x},\bm{x}') = (1+\alpha_1\|\bm{x}\|_2^2 + \alpha_1\|\bm{x}'\|_2^2)^{-1} \exp(-(2\alpha_2^2)^{-1}\|\bm{x}-\bm{x}'\|_2^2)$ for some $\alpha_1,\alpha_2>0$.
\end{remark}
We can now state our main result:
\begin{theorem} \label{sun proof}
Assume (A1,2',3,4',5) and take a RLS estimate with $\lambda = O(m^{-1/2})$.
When $\mathcal{D}_0$ are IID samples from $\pi$, the estimator $\hat{\mu}(\mathcal{D}_0,\mathcal{D}_1;f)$ is an unbiased estimator of $\mu(f)$ with $\mathbb{E}_{\mathcal{D}_0}\mathbb{E}_{\mathcal{D}_1}[(\hat{\mu}(\mathcal{D}_0,\mathcal{D}_1;f) - \mu(f))^2] = O(n^{-7/6})$.
\end{theorem}

CFs based on RLS therefore improve upon the Monte Carlo rate.
The hypotheses on $\pi$ are weak, only requiring that $\pi$ be continuously differentiable.
Empirical evidence (below) indicates stronger rates hold in more regular examples.
Indeed we can prove sharper results under stronger conditions that include boundedness of $\Omega$.
Details are reserved for a future publication \citep{Oates7}.

Importantly, the RLS estimate leads to a convenient closed-form expression for the CF estimator:
\begin{lemma} \label{explicit formulae}
Assume (A1,3).
The CF estimator based on RLS is
\begin{eqnarray}
\hat{\mu}(\mathcal{D}_0,\mathcal{D}_1;f) = \underbrace{\frac{1}{n-m} \bm{1}^T (\bm{f}_1 - \hat{\bm{f}}_1)}_{(*)} + \underbrace{\frac{\bm{1}^T (\bm{K}_0 + \lambda m \bm{I})^{-1} \bm{f}_0}{1 + \bm{1}^T (\bm{K}_0 + \lambda m \bm{I})^{-1} \bm{1}}}_{(**)}  \label{rewrite}
\end{eqnarray}
where $\bm{f}_0 = [f(\bm{x}_{1}),\dots,f(\bm{x}_m)]^T$, $\bm{f}_1 = [f(\bm{x}_{m+1}),\dots,f(\bm{x}_n)]^T$, $\bm{1} = [1,\dots,1]^T$, $(\bm{K}_0)_{i,j} = k_0(\bm{x}_{i},\bm{x}_{j})$ and the vector
\begin{eqnarray*}
\hat{\bm{f}}_1 := \bm{K}_{1,0} (\bm{K}_0 + \lambda m \bm{I})^{-1} \bm{f}_0 + (\bm{1} - \bm{K}_{1,0} (\bm{K}_0 + \lambda m \bm{I})^{-1}\bm{1}) \left( \frac{\bm{1}^T (\bm{K}_0 + \lambda m \bm{I})^{-1}\bm{f}_0}{1 + \bm{1}^T (\bm{K}_0 + \lambda m \bm{I})^{-1} \bm{1}} \right) 
\end{eqnarray*}
contains predictions for $\bm{f}_1$ based only on $\mathcal{D}_0$, with $(\bm{K}_{1,0})_{i,j} = k_0(\bm{x}_{m+i},\bm{x}_{j})$.
\end{lemma}
\begin{remark}
The estimator is a weighted combination of function values $\bm{f} = [\bm{f}_0^T,\bm{f}_1^T]^T$ with weights summing to one.
Estimates are readily obtained using standard matrix algebra.
Moreover the weights are independent of the test function $f$ and can be re-used to estimate multiple expectations $\mu(f_j)$ for a collection $\{f_j\}$.
\end{remark}
\begin{remark}
The samples $\mathcal{D}_1$ enter only through the term $(*)$ in Eqn. \ref{rewrite}, which vanishes in probability as $m \rightarrow \infty$.
Thus any randomness due to $\mathcal{D}_1$ vanishes and this gives another perspective on the source of super-root-$n$ convergence of the estimator.
\end{remark}
\begin{remark}
The term $(**)$ in Eqn. \ref{rewrite} is algebraically equivalent to BMC based on $\mathcal{H}_+$.
i.e. $(**)$ is the posterior mean for $\mu(f)$ based on a Gaussian process (GP) prior $f \sim \mathcal{GP}(0,k_+)$ and data $\mathcal{D}_0$ \citep[][Eqn. 9]{Rasmussen}.
Our general construction in Lemma \ref{explicit formulae} therefore ``heals'' BMC in the sense of (i) de-biasing the BMC estimator, (ii) generalising BMC to un-normalised densities and (iii) remaining agnostic to statistical paradigm (e.g. frequentist vs. Bayesian).
\end{remark}

\subsubsection{Non-asymptotic bounds} \label{preasy sec}

The naive computational complexity associated with the RLS estimate is $O(m^3)$ due to the solution of an $m \times m$ linear system.
In situations where $\bm{X}$ is expensive to simulate or $f$ is expensive to evaluate, $m$ is necessarily small and this additional computational cost will be negligible relative to model-based computation.
In such scenarios we are more interested in non-asymptotic behaviour:
\begin{theorem} \label{preasy}
Assume (A1,2',3,5). Let $\lambda \searrow 0$ to simplify presentation. Then 
\begin{eqnarray*}
|\hat{\mu}(\mathcal{D}_0,\mathcal{D}_1;f) - \mu(f) | \leq D(\mathcal{D}_0,\mathcal{D}_1)^{1/2} \|f\|_{\mathcal{H}_+}
\end{eqnarray*}
where
\begin{eqnarray*}
D(\mathcal{D}_0,\mathcal{D}_1)  = \frac{1}{(n-m)^2} \left[\frac{(\bm{1}^T \bm{K}_{1,0} \bm{K}_0^{-1} \bm{1})^2}{1 + \bm{1}^T\bm{K}_0^{-1}\bm{1}} - \bm{1}^T\bm{K}_{1,0}\bm{K}_0^{-1} \bm{K}_{0,1}\bm{1} + \bm{1}^T\bm{K}_1\bm{1} \right].
\end{eqnarray*}
Here $(\bm{K}_1)_{i,j} = k_0(\bm{x}_{m+i},\bm{x}_{m+j})$ and $\bm{K}_{0,1} = \bm{K}_{1,0}^T$.
\end{theorem}

Theorem \ref{preasy} provides an explicit error bound for $f \in \mathcal{H}_+$, mimicking the approach of (R)QMC \citep[][Sec. 2.3.3]{Dick}.
This offers a principled approach to selection of the design points $\mathcal{D}_0$ since, writing $\mathbb{V}_{\mathcal{D}_0,\mathcal{D}_1}$ for the variance with respect to the joint distribution of $\mathcal{D}_0$ and $\mathcal{D}_1$, we have
\begin{eqnarray*}
\mathbb{V}_{\mathcal{D}_0,\mathcal{D}_1}[\hat{\mu}(\mathcal{D}_0,\mathcal{D}_1;f)] \leq \mathbb{E}_{\mathcal{D}_0}\mathbb{E}_{\mathcal{D}_1}[D(\mathcal{D}_0,\mathcal{D}_1)] \|f\|_{\mathcal{H}_+}^2.
\end{eqnarray*}
In the extreme case where $\bm{x} \neq \bm{x}' \implies k_0(\bm{x},\bm{x}') = 0$, the discrepancy $D(\mathcal{D}_0,\mathcal{D}_1)$ reduces to $(n-m)^{-1}$ and we recover the usual root-$n$ rate.
A similar bound forms the basis for recent work on two-sample testing by \cite{Chwialkowski,Liu}.

\subsubsection{Implementation}

Several randomly chosen splits of the samples $\mathcal{D}$ into subsets $\mathcal{D}_0$ and $\mathcal{D}_1$ may be averaged over to reduce estimator variance.
We note that a multi-splitting estimator remains unbiased.
As an alternative to multi-splitting, for applications where consistency suffices and unbiased estimation is not essential, we also propose the simplified estimator $(**)$ in Eqn. \ref{rewrite} with $m = n$.
Empirical results below show that bias is negligible for practical purposes and, to pre-empt our conclusions, we recommend this simplified estimator for use in applications due to its reduced variance compared to the multi-splitting estimator.
In all cases the regularisation parameter $\lambda$ was taken to be the smallest power of 10 such that the kernel matrix $\bm{K}_0 + \lambda\bm{I}$ has condition number lower than $10^{10}$.

A kernel $k(\bm{x},\bm{x}';\bm{\alpha})$ typically involves hyper-parameters $\bm{\alpha}$ that must be specified.
Selection of $\bm{\alpha}$ can proceed via cross-validation, under the assumption that $\mathcal{D}_0$ are independent samples from $\pi$.
Specifically, we randomly split the samples $\mathcal{D}_0$ into $m'$ training samples $\mathcal{D}_{0,0}$ and $m-m'$ test samples $\mathcal{D}_{0,1}$.
Then we propose to select $\bm{\alpha}$ to minimise $\|\bm{f}_{(0,1)} - \hat{\bm{f}}_{(0,1)}\|_2$ where $\bm{f}_{(0,1)}$ is a vector of values $f_i$ for $\bm{x}_i \in \mathcal{D}_{0,1}$, and $\hat{\bm{f}}_{(0,1)}$ are the corresponding predicted values.
In this way we are targeting the EFAE that reflects the variance of the CF estimator.
We emphasise that the cross-validated estimator does not require additional sampling and will remain unbiased provided that preliminary cross-validation is performed only using $\mathcal{D}_0$.
In this paper we employed the kernel defined in Remark \ref{sup condition}, with hyper-parameters $\alpha_1,\alpha_2$.
Full pseudocode is provided in the supplement.

\subsection{Illustration} \label{ill sec}

To illustrate the method we begin with simple, tractable examples.
Consider the synthetic problem of estimating the expectation of $f(\bm{X}) = \sin(\frac{\pi}{d}\sum_{i=1}^d X_i)$ where $\bm{X}$ is a $d$-dimensional standard Gaussian random variable.
By symmetry the true expectation is $\mu(f) = 0$.
Initially we take $n=50$ IID samples and consider the scalar case $d=1$.
Cross-validation was used to select tuning parameters.
Specifically:
(i) We selected the hyper-parameters $\alpha_1 = 0.1$, $\alpha_2=1$ on the basis that this approximately minimised the cross-validation error (Fig. S1a).
(ii) We found that estimator variance due to sample-splitting was minimised when at least half of the samples were allocated to $\mathcal{D}_0$ (Fig. S1b).
We therefore set a conservative default $m = \lceil n/2 \rceil$.
(iii) Empirical results showed that little additional variance reduction occurs from employing multiple splits (Fig. S1c), so we chose to just use a single split.
(iv) Finally, we found that the bias of the simplified estimator was negligible ($<\sim 10^{-3}$) compared to Monte Carlo error ($\sim 10^{-2}$) (Fig. S1d).
This is in line with an analogous result for classical control variates, where estimator bias vanishes asymptotically with respect to Monte Carlo error \citep[][p.200]{Glasserman}.

\begin{figure}[t!]
\centering
\includegraphics[width = \textwidth]{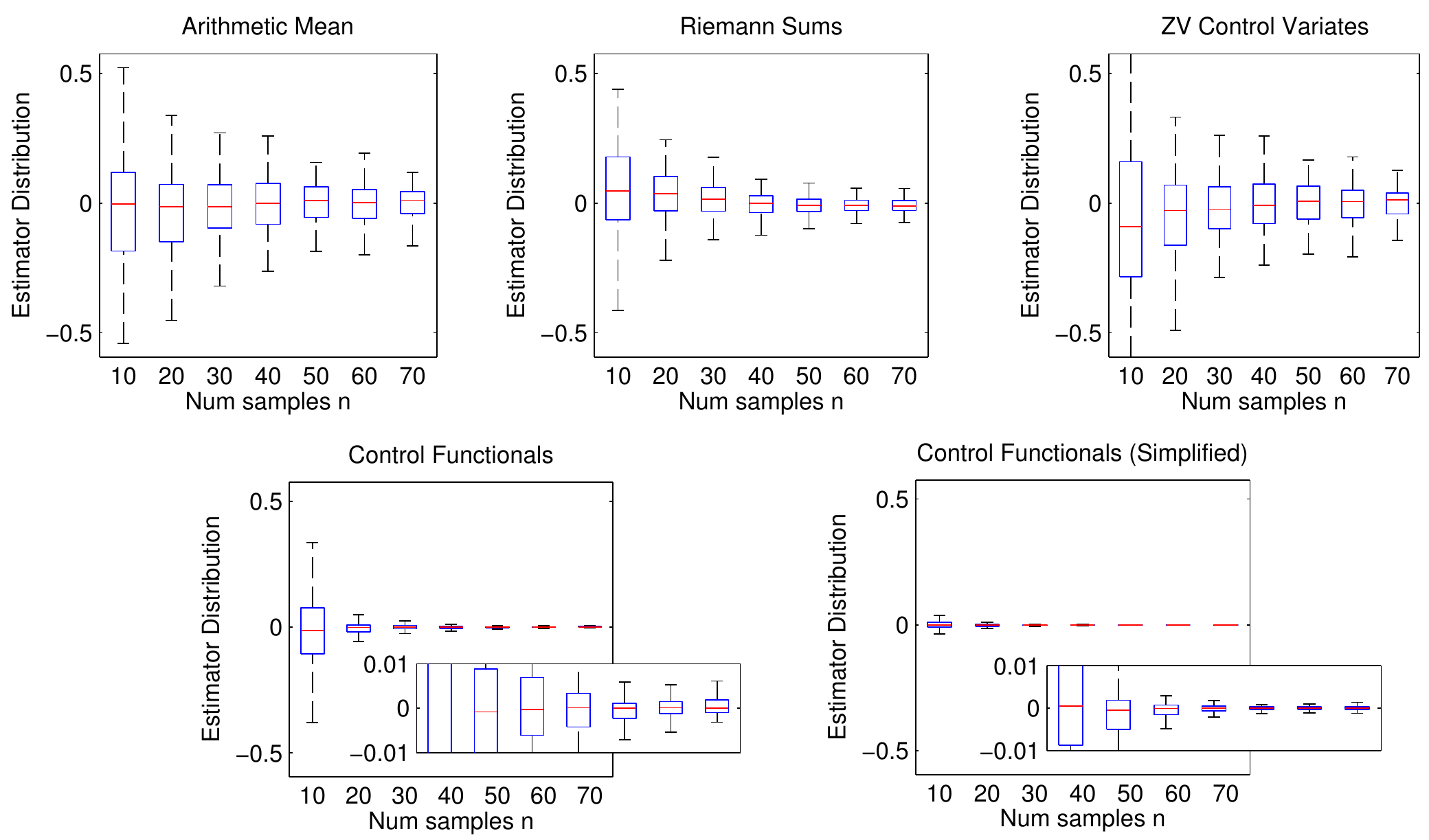}
\caption{Illustration on a synthetic integration problem in $d=1$ dimension.
Here we display the empirical sampling distribution of Monte Carlo estimators, based on $n$ samples and 100 independent realisations.
[The settings for all methods were as described in the main text.]}
\label{box1}
\end{figure}

In Fig. \ref{box1} we summarise the sampling distribution of both the sample-splitting and simplified CF estimators as the number of samples $n$ is varied.
The alternative approaches of the arithmetic mean, Riemann sums and ``zero variance'' (ZV) control variates are also shown, the latter being based on quadratic polynomials \citep{Mira}.
It is visually apparent that CFs enjoy the lowest variance at all samples sizes considered.
We note that in this synthetic example, where there are essentially no computational restrictions, the CF framework is unnecessary and gains in precision come with comparable increases in computational cost.
However we emphasise that, in the serious applications that follow, the CF calculations requires negligible computational resources in comparison to simulation from the model.
Super-root-$n$ convergence could perhaps be achieved by employing polynomials of increasing degree in the ZV method, but our implementation of this approach did not provide stable estimates in this example (full details in the supplement).

\begin{figure}[t!]
\centering
\includegraphics[width = 0.9\textwidth]{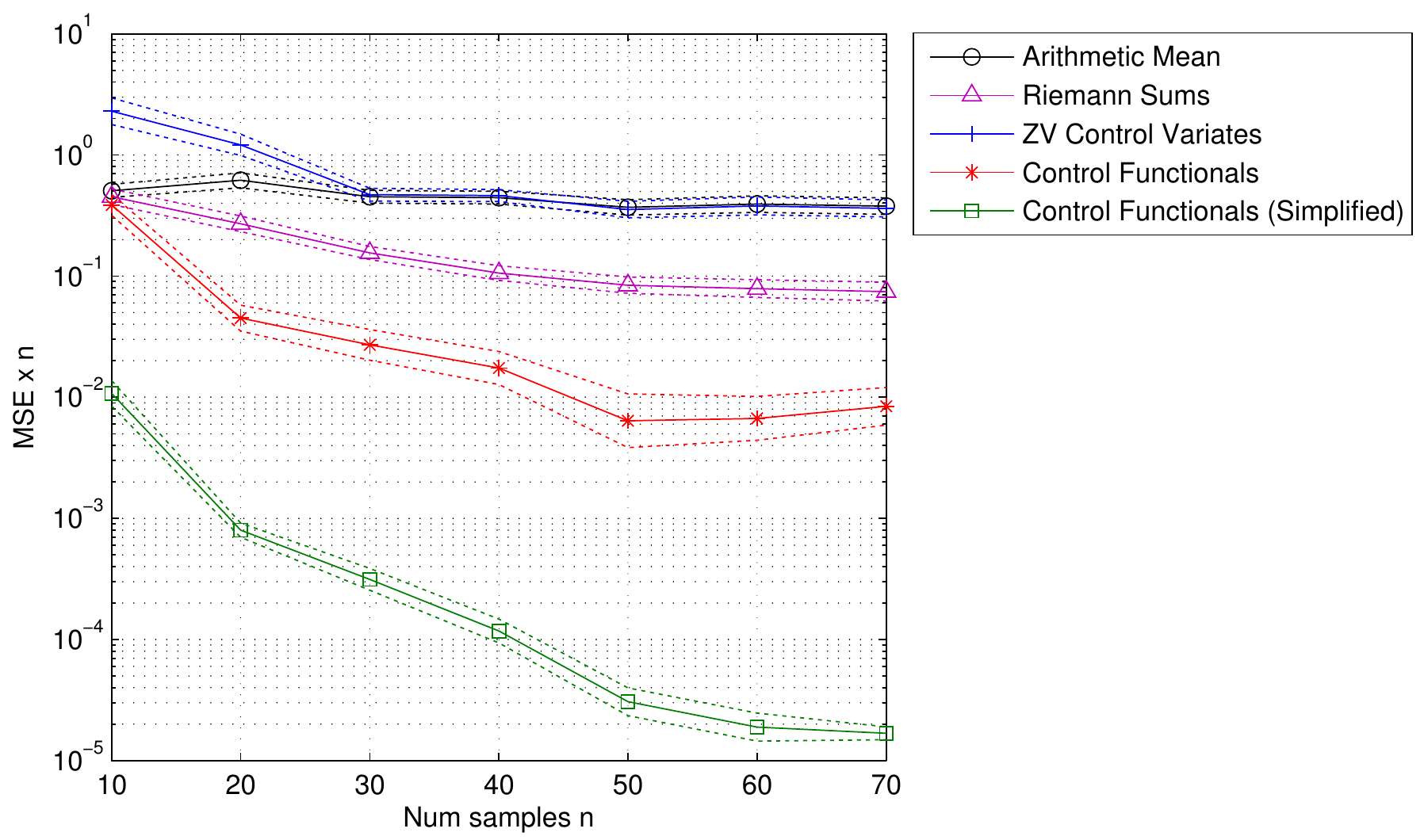}
\caption{Illustration on a synthetic integration problem in $d=1$ dimension (continued).
Empirical assessment of asymptotic properties.
\label{line1}
}
\end{figure}

Since the performance of CF is so pronounced, in order to more clearly visualise the results for all sample sizes, in Fig. \ref{line1} we plot the estimator mean square error (MSE) scaled by $n$, so that root-$n$ convergence corresponds to a horizontal line.
Empirical results here are consistent with theory, showing that the arithmetic mean and control variates all achieve a constant factor variance reduction, whereas Riemann sums and CFs achieve super-root-$n$ convergence.
In this example CFs significantly outperformed Riemann sums, the latter being based on piecewise linear approximations.
We plot results for both the sample-splitting CF estimator and the simplified CF estimator, observing that the latter has lower variance.

To assess the generality of our conclusions we considered going beyond the scalar case to examples with dimensions $d=3$ and $d=5$.
The analogous results in Figs. S2a, S2b show that, whilst increasing dimensionality presents fundamental challenges for all the variance reduction methods, CF continues to out-perform alternatives.
Going further we considered a variety of alternative problems, varying both the test function $f$ and the density $\pi$.
These include several pathological cases, with results summarised in Table S1.
The results marked (b) echo the conclusions of \cite{Mira}, that ZV control variates are effective in many cases where $f$ is well-approximated by a low-degree polynomial and $\pi$ is a Gaussian or gamma density.
However, when $f$ is not well-approximated by a low-degree polynomial, or when $\pi$ takes a more complex form, as in cases marked (c), ZV control variates can be outperformed by CFs, which have the potential to decrease variance dramatically.
We then investigated how CFs can fail when theoretical assumptions are violated (see examples marked ``CF $\times$'').
As expected, violation of (A2) and (A5) in (e), (g) respectively led to poor performance of the CF estimator.
Interestingly, violation of differentiability in example (f) did not lead to poor estimation, though this may be because $\pi$ was only non-differentiable at a single point.

We have not reported computational times for these experiments.
Our work is motivated by settings in which either simulation from $\pi$ or evaluation of $f$ (or both) are computationally prohibitive, so that additional effort required to implement CFs is negligible by comparison; we illustrate this with two realistic applications the next section.

\section{Applications}

Two applications are considered that together present many of the challenges associated with complex models.
Firstly we consider marginalisation of hyper-parameters in hierarchical models, focussing on a non-parametric prediction problem.
Here evaluation of $f$ forms a computational bottleneck due to the required inversion of a large matrix.
For this problem, CFs are shown to offer significant computational savings.
Secondly we consider computation of normalising constants for models based on non-linear ordinary differential equations (ODEs).
Evaluation of the likelihood function requires numerical integration of a system of ODEs and dominates computational expenditure in both sampling from $\pi$ and evaluation of $f$.
Here CFs combine with gradient-based population MCMC and thermodynamic integration in order to deliver a state-of-the-art technique for low-variance estimation of normalising constants.

\subsection{Marginalisation in hierarchical models}

A fully Bayesian treatment of hierarchical models aims to marginalise over hyper-parameters, but this often entails a prohibitive level of computation.
Here we explore the efficacy of CFs in such situations.

\subsubsection{A hierarchical GP model} \label{hierarchical GP}

The marginalisation of hyper-parameters is a common problem in spatial statistics and Bayesian statistics in general \citep{Besag,Agapiou,Filippone}.
Here we consider one such model that is based on $p$-dimensional GP regression.
Denote by $Y_i \in \mathbb{R}$ a measured response variable at state $\bm{z}_i \in \mathbb{R}^p$, assumed to satisfy $Y_i = g(\bm{z}_i) + \epsilon_i$ where $\epsilon_i \sim N(0,\sigma^2)$ are independent for $i = 1,\dots,N$ and $\sigma > 0$ will be assumed known.
In order to use training data $(y_i,\bm{z}_i)_{i=1}^n$ to make predictions regarding an unseen test point $\bm{z}_*$, we place a GP prior $g \sim \mathcal{GP}(0,c(\bm{z},\bm{z}';\bm{\theta}))$ where $c(\bm{z},\bm{z}';\bm{\theta}) = \theta_1 \exp (- \frac{1}{2\theta_2^2} \|\bm{z} - \bm{z}'\|_2^2 )$.
Here $\bm{\theta} = (\theta_1,\theta_2)$ are hyper-parameters that control how training samples are used to predict the response at a new test point.
In the fully-Bayesian framework these are assigned hyper-priors, say $\theta_1 \sim \Gamma(\alpha,\beta)$, $\theta_2 \sim \Gamma(\gamma,\delta)$ in the shape/scale parametrisation, which we write jointly as $\pi(\bm{\theta})$.

\subsubsection{Marginalising the GP hyper-parameters} \label{HGP}

We are interested in predicting the value of the response $Y_*$ corresponding to an unseen state vector $\bm{z}_*$.
Our estimator will be the Bayesian posterior mean given by
\begin{eqnarray}
\hat{Y}_* := \mathbb{E}[Y_*|\bm{y}] = \int \mathbb{E}[Y_*|\bm{y},\bm{\theta}] \pi(\bm{\theta}) \mathrm{d}\bm{\theta}, \label{robot target}
\end{eqnarray}
where we implicitly condition on the covariates $\bm{z}_1,\dots,\bm{z}_N,\bm{z}_*$.
Eqn. \ref{robot target} is unavailable in closed form and we therefore naive a Monte Carlo estimate by sampling $\bm{\theta}_1,\dots,\bm{\theta}_n$ independently from the prior $\pi(\bm{\theta})$ (more efficient QMC estimates are considered later). 
Phrasing in terms of our previous notation, the function of interest is 
\begin{eqnarray*}
f(\bm{\theta}) = \mathbb{E}[Y_*|\bm{y},\bm{\theta}] = \bm{C}_{*,N} (\bm{C}_N + \sigma^2 \bm{I}_{N \times N})^{-1} \bm{y}
\end{eqnarray*}
where $(\bm{C}_N)_{i,j} = c(\bm{z}_i,\bm{z}_j;\bm{\theta})$ and $(\bm{C}_{*,N})_{1,j} = c(\bm{z}_*,\bm{z}_j;\bm{\theta})$ and the underlying distribution is $\pi(\bm{\theta})$.
Each evaluation of the integrand $f(\bm{\theta})$ requires $O(N^3)$ operations due to the matrix inversion; this can be reduced by employing a ``subset of regressors'' approximation
\begin{eqnarray}
f(\bm{\theta}) \approx \bm{C}_{*,N'} (\bm{C}_{N',N} \bm{C}_{N,N'} + \sigma^2 \bm{C}_{N'})^{-1} \bm{C}_{N',N} \bm{y} \label{sor approx}
\end{eqnarray}
where $N' < N$ denotes a subset of the full data \citep[see Sec. 8.3.1 of][for full details]{Rasmussen2}.
To facilitate the illustration below, which investigates the sampling distribution of estimators, we take a random subset of $N=1,000$ training points and a subset of regressors approximation with $N' = 100$.
However we emphasise that evaluation of Eqn. \ref{sor approx} will typically be based on much larger $N$ and $N'$ and will be extremely expensive in general.
In applications we would therefore have to proceed with Monte Carlo estimation based on only a small number $n$ of these function evaluations.

\subsubsection{SARCOS robot arm} \label{sarcos results}

We used the hierarchical GP model in Sec. \ref{HGP} to estimate the inverse dynamics of a seven degrees-of-freedom SARCOS anthropomorphic robot arm. 
The task, as described in \citet[][Sec. 8.3.1]{Rasmussen2}, is to map from a 21-dimensional input space (7 positions, 7 velocities, 7 accelerations) to the corresponding 7 joint torques using the hierarchical GP model described in Sec. \ref{hierarchical GP}. 
Following \cite{Rasmussen2} we present results below on just one of the mappings, from the 21 input variables to the first of the seven torques.
The dataset consists of 48,933 input-output pairs, of which 44,484 were used as a training set and the remaining 4,449 were used as a test set. 
The inputs were translated and scaled to have mean zero and unit variance on the training set. The outputs were centred so as to have mean zero on the training set.
Here  $\sigma = 0.1$, $\alpha = \gamma = 25$, $\beta = \delta = 0.04$, so that each hyper-parameter $\theta_i$ has a prior mean of $1$ and a prior standard deviation of $0.2$.

For each test point $\bm{z}_*$ we estimated the sampling standard deviation of $\hat{Y}_*$ over 10 independent realisations of the Monte Carlo sampling procedure.
For CF we took default hyper-parameters $\alpha_1 = 0.1$, $\alpha_2 = 1$, the latter reflecting the fact that the training data were standardised.
The estimator standard deviations were estimated in this way for all 4,449 test samples and the full results are shown in Fig. \ref{robot}.
Note that each test sample corresponds to a different function $f$ and thus these results are quite objective, encompassing thousands of different Monte Carlo integration problems.
Results show that, for the vast majority of integration problems, CF achieves a lower estimator variance compared with both the arithmetic mean estimator and ZV control variates.
Here the cost of post-processing the Monte Carlo samples (using either ZV control variates or CF) is negligible in comparison to the cost of evaluating the function $f$, even once.
Indeed, CF requires that we invert a $n \times n$ matrix once, where $n$ is no larger than $N'$ in this example.

In the supplement we investigate an extension that draws design points $\mathcal{D}_0$ using RQMC.
Results show that CFs+RQMC outperforms RQMC alone.

\begin{figure}[t!]
\includegraphics[width = \textwidth,clip,trim = 0cm 0cm 0cm 0cm]{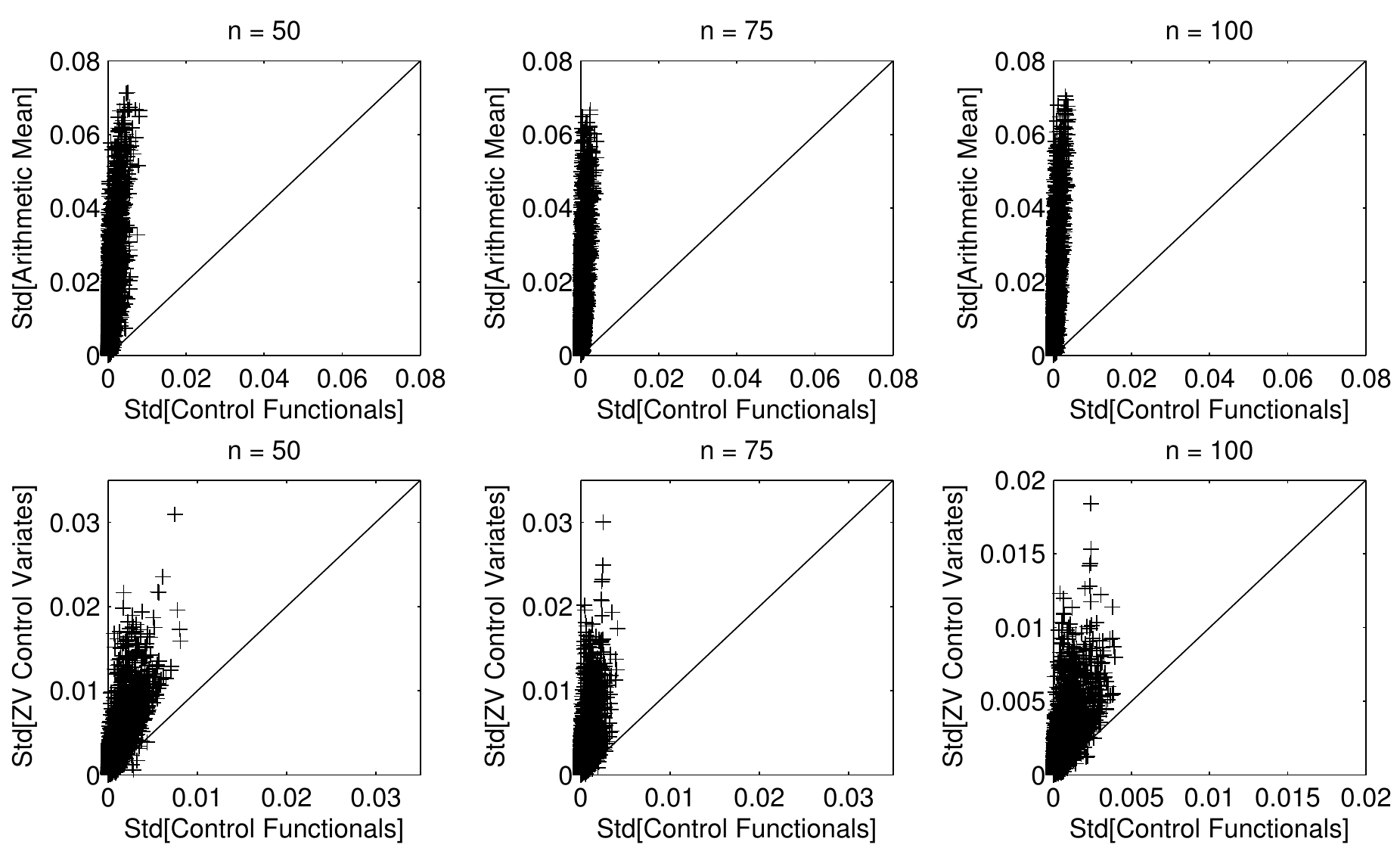}
\caption{Marginalisation of hyper-parameters in hierarchical models. [Here we display the sampling standard deviation of Monte Carlo estimators for the posterior predictive mean $\mathbb{E}[Y_*|\bm{y}]$ in the SARCOS robot arm example, computed over 10 independent realisations. Each point, representing one Monte Carlo integration problem, is represented by a cross.]}
\label{robot}
\end{figure}

\subsection{Normalising constants for non-linear ODE models}

Our second application concerns the estimation of normalising constants for non-linear ODE models \citep[e.g.][]{Calderhead}.
Recent empirical investigations recommend thermodynamic integration (TI) for this task \citep[e.g.][]{Friel3}.
The control variate method of \cite{Mira2} was recently applied to TI by \cite{Oates2}, who found that this ``controlled thermodynamic integral'' (CTI) was extremely effective for standard regression models, but only moderately effective in complex models including non-linear ODEs due to poor approximation by low-degree polynomials.
Below we study the application of CFs to TI in this setting where CTI is less effective.

\subsubsection{Thermodynamic integration}

Conditional on an inverse temperature parameter $t$, the ``power posterior'' for parameters $\bm{\theta}$ given data $\bm{y}$ is defined as $p(\boldsymbol{\theta} | \bm{y},t) \propto p(\bm{y}|\boldsymbol{\theta})^tp(\boldsymbol{\theta})$ \citep{Friel}.
Varying $t \in [0,1]$ produces a continuous path between the prior $p(\bm{\theta})$ and the posterior $p(\bm{\theta}|\bm{y})$ and it is assumed here that all intermediate distributions exist and are well-defined.
The standard thermodynamic identity is
\begin{eqnarray}
\log p(\bm{y}) =  \int_0^1 \mathbb{E}_{\bm{\theta}|\bm{y},t} [\log p(\bm{y}|\bm{\theta})] \mathrm{d}t,
\label{ML}
\end{eqnarray}
where the expectation in the integrand is with respect to the power posterior.
In TI, the one-dimensional integral in Eqn. \ref{ML} is evaluated numerically using a quadrature approximation over a discrete temperature ladder $0 = t_0 < t_1 < \dots < t_m = 1$.
Here we use the second-order quadrature recommended by \cite{Friel2}:
\begin{eqnarray*}
\log p(\bm{y}) \approx \sum_{i = 0}^{m-1} \frac{(t_{i+1}-t_i)}{2} (\hat{\mu}_i + \hat{\mu}_{i+1}) - \frac{(t_{i+1}-t_i)^2}{12} (\hat{\nu}_{i+1} - \hat{\nu}_i), 
\end{eqnarray*}
where $\hat{\mu}_i$, $\hat{\nu}_i$ are Monte Carlo estimates of the posterior mean and variance respectively of $\log p(\bm{y}|\bm{\theta})$ when $\bm{\theta}$ arises from $p(\bm{\theta}|\bm{y},t_i)$.
CTI uses ZV control variates to reduce the variance of these estimates.
However, in complex models $\log p(\bm{y}|\bm{\theta})$ will be poorly approximated by a low-degree polynomial and $p(\bm{\theta}|\bm{y},t)$ will be non-Gaussian; this explains the mediocre performance of CTI in these cases.
In contrast, CFs should still be able to deliver gains in estimation.

\subsubsection{Non-linear ODE models}

\begin{figure}[t!]
\centering
\includegraphics[clip,trim = 0cm 11cm 0cm 10.5cm,width = \textwidth]{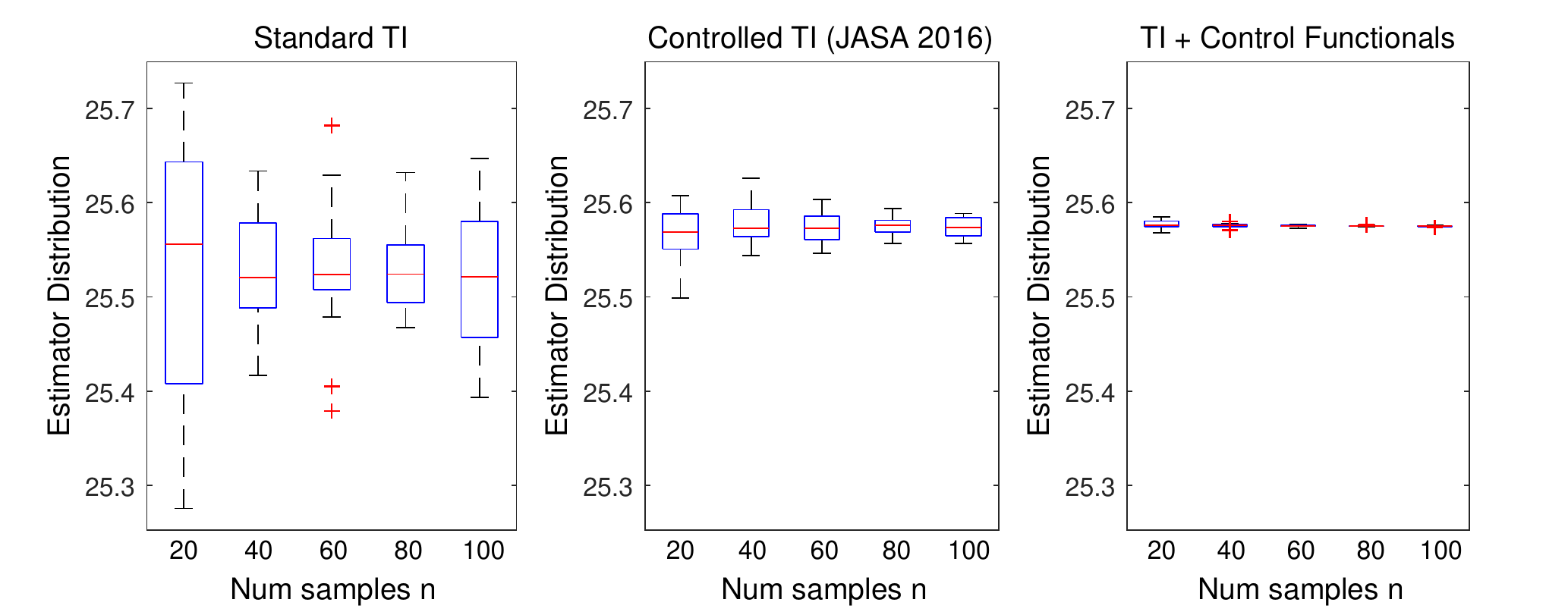}
\caption{Estimation of normalising constants for non-linear ordinary differential equations using thermodynamic integration (TI); van der Pol oscillator example.
[Here we show the distribution of 100 independent realisations of each estimator for $\log p(\bm{y})$.
``Standard TI'' is based on arithmetic means.
``Controlled TI'' is based on ZV control variates.]}
\label{ODE}
\end{figure}

The approach is illustrated by computing the marginal likelihood for a non-linear ODE model (the van der Pol oscillator), described in full in the supplement.
For TI, a temperature schedule $t_i = (i/30)^5$ was used, following the recommendation by \cite{Calderhead}.
The power posterior is not available in closed form, precluding the straight-forward generation of IID samples.
Instead, samples from each of the power posteriors $p(\bm{\theta}|\bm{y},t_i)$ were obtained using population MCMC, involving both (i) ``within-temperature'' proposals produced by the (simplified) m-MALA algorithm of \cite{Girolami}, and (ii) ``between-temperature'' proposals, as described previously by \cite{Calderhead}.
Gradient information is thus pre-computed in the sampling scheme and can be leveraged ``for free'', as noted by \cite{Papamarkou}.
We denote the number of samples by $n$, such that for each of the 31 temperatures we obtained $n$ samples (a total of $31 \times n$ occasions where the system of ODEs was integrated numerically).
Both sampling and evaluation of the integrand are computationally expensive, requiring the numerical solution of a system of ODEs.

Results in Fig. \ref{ODE} show that the CTI estimator improves upon the standard TI estimator, but a more substantial reduction in estimator variance results from using the CF method.
For the CF computation we have used the simplified but biased CF estimator, since TI in any case produces a biased estimate for the normalising constant due to numerical quadrature.
The hyper-parameters $\alpha_1 = 0.1$, $\alpha_2 = 3$ were selected on the basis of cross-validation.
The additional cost of using CF is essentially zero relative to running the population MCMC sampler, the latter requiring repeated solution of the ODE system.

\section{Discussion}

This paper developed a novel and general approach to integration that achieves super-root-$n$ convergence.
An important feature of CFs is that variance reduction is formulated as a {\it post-hoc} step.
This has several advantages:
(i) No modification is required to existing computer code associated with either the sampling process or the model itself.
(ii) Specific implementational choices, e.g. for the kernel, can be made {\it after} performing expensive simulations.
Through exploitation of recent results in functional analysis we were able to realise our general framework and construct estimators with an analytic form.
Empirical results evidenced the practical utility of CF estimators in settings where gradient information is available and the dimensionality of the problem is not too large (e.g. $\leq 10$).
The paper concludes below by suggesting directions for further research.

In terms of methodology:
(i) The estimates we presented here are not parameterisation-invariant. Likewise the specification of $f$ and $\pi$ is not unique, as we can employ an importance sampling transformation $f \mapsto (f \pi) / \pi'$, $\pi \mapsto \pi'$. It would therefore be interesting to elicit effective parametrisations as an additional {\it post-hoc} step.
(ii) The version of CFs presented here is limited in terms of the dimension of the problems for which it is effective.
Techniques for high-dimensional functional approximation should be applicable in the context of CFs \citep[e.g.][]{Dick2} and this forms part of our ongoing research.

In terms of theory:
(i) For bounded $\Omega$, sharper asymptotics are provided in a sequel, \cite{Oates7}.
These account for various levels of smoothness of both $f$ and $\pi$ and help to explain the strong empirical results presented here.
However the case of unbounded $\Omega$ seems considerably more challenging to characterise.
(ii) For problems involving un-normalised densities $\pi$, sampling is naturally facilitated by MCMC.
The analysis of CFs is carried out in \cite{Oates7} under a uniform ergodicity assumption.
For unbounded $\Omega$ this condition is too strong and future work will aim to relax this constraint.

In terms of application:
(i) Our methods were motivated by the un-normalised densities arising in Bayesian computation.
An extension should be possible to models with unknown, parameter-dependent normalising constants, which include e.g. Markov random fields \citep{Everitt} and random network models \citep{Friel4}.
(ii) An interesting direction would be to use the discrepancy $D(\mathcal{D}_0,\mathcal{D}_1)$ as a tool for assessment of MCMC convergence, providing a reproducing kernel Hilbert space alternative to \cite{Gorham}.

Finally we note that \cite{Oates3} provide a complementary study of CF strategies in the QMC setting.

\paragraph{Acknowledgements:}
The authors are grateful to the editor and referees, whose valuable feedback helped to improve the paper.
The authors benefited from discussions with Sergios Agapiou, Michel Caffarel, Adam Johansen, Christian Robert, Daniel Simpson and Tim Sullivan.
CJO was supported by EPSRC [EP/D002060/1] and the ARC Centre for Excellence in Mathematical and Statistical Frontiers. 
MG was supported by EPSRC [EP/J016934/1], EU [EU/259348], an EPSRC Established Career Fellowship and a Royal Society Wolfson Research Merit Award.
NC was supported by the ANR (Agence Nationale de la Recherche) grant Labex ECODEC ANR [11-LABEX-0047].

\appendix

\section{Proofs} \label{aproofs}

\begin{proof}[Proof of Proposition \ref{theo1}]
We exploit the unbiasedness property $\mathbb{E}_{\mathcal{D}_1}[\hat{\mu}(\mathcal{D}_0,\mathcal{D}_1;f) - \mu(f)] = 0$ to show that
\begin{eqnarray*}
\mathbb{E}_{\mathcal{D}_0} \mathbb{E}_{\mathcal{D}_1}[(\hat{\mu}(\mathcal{D}_0,\mathcal{D}_1;f) - \mu(f))^2] & = & \mathbb{E}_{\mathcal{D}_0} \mathbb{V}_{\mathcal{D}_1} [\hat{\mu}(\mathcal{D}_0,\mathcal{D}_1;f) ] \\
& = & (n-m)^{-1} \mathbb{E}_{\mathcal{D}_0} [\sigma^2(f - s_{f,\mathcal{D}_0})]
\end{eqnarray*}
where $(n-m)^{-1} = O(n^{-1})$ and by hypothesis $\mathbb{E}_{\mathcal{D}_0} [\sigma^2(f - s_{f,\mathcal{D}_0})] = O(m^{-\delta})$.
Thus using $m = O(n^{\gamma})$ produces an overall rate $O(n^{-1-\gamma\delta})$, as required.
\end{proof}

\begin{proof}[Proof of Proposition \ref{div prop}]
(A1) ensures $\psi$ is well-defined.
Since $\partial\Omega$ is piecewise smooth we can apply the divergence theorem \citep[e.g.][p.159]{Bourne} to obtain
\begin{eqnarray*}
\mu(\psi) = \int_{\Omega} \psi(\bm{x}) \pi(\bm{x}) \mathrm{d}\bm{x} = \int_{\Omega} \nabla_{\bm{x}} \cdot [\bm{\phi}(\bm{x}) \pi(\bm{x})] \mathrm{d}\bm{x} = \oint_{\partial\Omega} \pi(\bm{x}) \bm{\phi}(\bm{x}) \cdot \bm{n}(\bm{x}) S(\mathrm{d}\bm{x}), 
\end{eqnarray*}
which is zero by (A2). 
The use of this identity in statistical applications is often attributed to \cite{Stein}.
Thus $\mu(s_{f,\mathcal{D}_0}) = \mu(c) + \mu(\psi) = c$, as required.
\end{proof}

\begin{proof}[Proof of Theorem \ref{conjugate}]
{\it Stage 1:} We begin by defining the set of CFs $\mathcal{H}_0$.
Given a reproducing kernel $k : \Omega \times \Omega \rightarrow \mathbb{R}$ for the reproducing kernel Hilbert space (RKHS) $\mathcal{H}$, define the canonical feature map $\Phi : \Omega \rightarrow \mathcal{H}$ by $\Phi(\bm{x}) = k(\cdot,\bm{x})$.
Under (A1) the gradient function $\bm{u} : \Omega \rightarrow \mathbb{R}$ is well-defined.
Under (A3) $k$ has mixed first order partial derivatives; it follows that all elements $\phi_i \in \mathcal{H}$ are differentiable and thus $(\partial/\partial x_i) \phi_i(\bm{x})$ is well-defined \citep[][Cor. 4.36, p131]{Steinwart2}.
We then have that 
\begin{eqnarray}
\psi(\bm{x}) & = & \sum_{i=1}^d (\partial/\partial x_i) \phi_i(\bm{x}) + u_i(\bm{x}) \phi_i(\bm{x}) \nonumber \\
& = & \sum_{i=1}^d (\partial/\partial x_i) \langle \phi_i , \Phi(\bm{x}) \rangle_{\mathcal{H}} + u_i(\bm{x}) \langle \phi_i , \Phi(\bm{x}) \rangle_{\mathcal{H}} \nonumber \\
& = & \sum_{i=1}^d \langle \phi_i , (\partial/\partial x_i)\Phi(\bm{x}) + u_i(\bm{x}) \Phi(\bm{x}) \rangle_{\mathcal{H}} = \sum_{i=1}^d \langle \phi_i , \Phi_i^*(\bm{x}) \rangle_{\mathcal{H}}, \nonumber
\end{eqnarray}
where we have used the notation $\Phi_i^*(\bm{x}) = (\partial/\partial x_i) \Phi(\bm{x}) + u_i(\bm{x}) \Phi(\bm{x})$.
Write $\Phi^* : \Omega^d \rightarrow \mathbb{R}^d$ for the derived feature map with $i$th component $\Phi_i^*$.
Define the set of all CFs $\psi$ of this form as
$$\mathcal{H}_0 = \{\psi : \Omega \rightarrow \mathbb{R} \text{ such that } \forall\bm{x}\in\Omega, \; \psi(\bm{x}) = \langle \bm{\phi} , \Phi^*(\bm{x}) \rangle_{\mathcal{H}^d} \text{ for some } \bm{\phi} \in \mathcal{H}^d \}.$$
Clearly $\mathcal{H}_0$ is a vector space with addition and multiplication defined pointwise; $(\lambda\psi + \lambda'\psi')(\bm{x}) = \lambda\psi(\bm{x}) + \lambda'\psi'(\bm{x})$.

{\it Stage 2:}
We now show that $\mathcal{H}_0$ can be endowed with the structure of a RKHS.
To this end, define a norm on $\mathcal{H}_0$ by
$$
\|\psi\|_{\mathcal{H}_0} := \inf_{\bm{\phi} \in \mathcal{H}^d} \{\|\bm{\phi}\|_{\mathcal{H}_d} \text{ such that } \forall\bm{x}\in\Omega, \; \psi(\bm{x}) = \langle \bm{\phi} , \Phi^*(\bm{x}) \rangle_{\mathcal{H}^d} \}.
$$
Theorem 4.21 (p121) of \cite{Steinwart2} immediately gives that the normed space $(\mathcal{H}_0,\|\cdot\|_{\mathcal{H}_0})$ is a RKHS whose kernel $k_0$ satisfies $k_0(\bm{x},\bm{x}') = \langle \Phi^*(\bm{x}) , \Phi^*(\bm{x}') \rangle_{\mathcal{H}^d}$.
Thus we can directly calculate
\begin{eqnarray*}
k_0(\bm{x},\bm{x}') & = & \langle \Phi^*(\bm{x}) , \Phi^*(\bm{x}') \rangle_{\mathcal{H}^d} \\
& = & \sum_{i=1}^d \langle (\partial/\partial x_i) \Phi(\bm{x}) + u_i(\bm{x}) \Phi(\bm{x}) , (\partial/\partial x_i') \Phi(\bm{x}') + u_i(\bm{x}') \Phi(\bm{x}') \rangle_{\mathcal{H}} \\
& = & \sum_{i=1}^d \begin{array}{ll} (\partial/\partial x_i) (\partial/\partial x_i') k(\bm{x},\bm{x}') + u_i(\bm{x}) (\partial/\partial x_i') k(\bm{x},\bm{x}') \\ \; \; \; \; \; \; \; \; \; \; \; \; \; \; \; \; \; \; \; \;   + u_i(\bm{x}') (\partial/\partial x_i) k(\bm{x},\bm{x}') + u_i(\bm{x}) u_i(\bm{x}') k(\bm{x},\bm{x}'), \end{array}
\end{eqnarray*}
where the interchange of derivative and inner product is justified by (A3) and Lemma 4.34 (p130) in \cite{Steinwart2}.
This completes the proof.
\end{proof}

\begin{proof}[Proof of Lemma \ref{mean ele lem}]
(A1,3) ensure the kernel $k_0$ is well-defined. 
Then
\begin{eqnarray*}
\int_{\Omega} k_0(\bm{x},\bm{x}') \pi(\bm{x}') \mathrm{d}\bm{x}' & = & \int_{\Omega} [\nabla_{\bm{x}} \cdot \nabla_{\bm{x}'} k(\bm{x},\bm{x}')] \pi(\bm{x}') \mathrm{d}\bm{x}' + \int_{\Omega} [\bm{u}(\bm{x}) \cdot \nabla_{\bm{x}'} k(\bm{x},\bm{x}')] \pi(\bm{x}') \mathrm{d}\bm{x}' \\
& & + \int_{\Omega} [\bm{u}(\bm{x}') \cdot \nabla_{\bm{x}} k(\bm{x},\bm{x}')] \pi(\bm{x}') \mathrm{d}\bm{x}' + \int_{\Omega} [\bm{u}(\bm{x}) \cdot \bm{u}(\bm{x}') k(\bm{x},\bm{x}')] \pi(\bm{x}') \mathrm{d}\bm{x}' \\
& = & \int_{\Omega} [\nabla_{\bm{x}'} \cdot \nabla_{\bm{x}} k(\bm{x},\bm{x}')] \pi(\bm{x}') + [\nabla_{\bm{x}} k(\bm{x},\bm{x}')] \cdot [\nabla_{\bm{x}'} \pi(\bm{x}')] \mathrm{d}\bm{x}' \\
& & + \bm{u}(\bm{x}) \cdot \int_{\Omega} [\nabla_{\bm{x}'} k(\bm{x},\bm{x}')] \pi(\bm{x}') + k(\bm{x},\bm{x}') [\nabla_{\bm{x}'} \pi(\bm{x}')] \mathrm{d}\bm{x}' \\
& = & \int_{\Omega} \nabla_{\bm{x}'} \cdot \left\{ [\nabla_{\bm{x}} k(\bm{x},\bm{x}')] \pi(\bm{x}') \right\} \mathrm{d}\bm{x}' + \bm{u}(\bm{x}) \cdot \int_{\Omega} \nabla_{\bm{x}'} \left\{ k(\bm{x},\bm{x}') \pi(\bm{x}) \right\} \mathrm{d}\bm{x}'.
\end{eqnarray*}
Now using the divergence theorem \citep[][p.159]{Bourne} we obtain
\begin{eqnarray*}
= \underbrace{\oint_{\partial\Omega} \nabla_{\bm{x}} k(\bm{x},\bm{x}') \pi(\bm{x}') \cdot \bm{n}(\bm{x}')  S(\mathrm{d}\bm{x}')}_{=0 \; \pi\text{-a.e. from (A2')}} + \bm{u}(\bm{x}) \cdot \underbrace{\oint_{\partial\Omega} k(\bm{x},\bm{x}') \pi(\bm{x}') \bm{n}(\bm{x}')  S(\mathrm{d}\bm{x}')}_{=\bm{0} \; \pi\text{-a.e. from (A2')}},
\end{eqnarray*}
proving the claim.
\end{proof}

\begin{proof}[Proof of Lemma \ref{finite variance}]
From Theorem \ref{conjugate}, (A1,3) ensure $\mathcal{H}_0$ is well-defined.
Moreover, from Lemma \ref{mean ele lem}, (A1,2',3) imply that $\mu(\psi) = 0$ and thus
$$
\sigma^2(\psi) = \int_{\Omega} \psi(\bm{x})^2 \pi(\bm{x}) \mathrm{d}\bm{x}.
$$
Now, given $\psi \in \mathcal{H}_0$, we need to show $\sigma^2(\psi) < \infty$.
By the reproducing property followed by the Cauchy-Schwarz inequality, we have
\begin{eqnarray*}
|\psi(\bm{x})| = |\langle \psi, k_0(\cdot,\bm{x}) \rangle_{\mathcal{H}_0}| \leq \|\psi\|_{\mathcal{H}_0} \|k_0(\cdot,\bm{x})\|_{\mathcal{H}_0}.
\end{eqnarray*}
Using the reproducing property again, we have $\|k_0(\cdot,\bm{x})\|_{\mathcal{H}_0}^2 = k_0(\bm{x},\bm{x})$ and it follows from (A4) that 
\begin{eqnarray*}
\sigma^2(\psi) = \int \psi(\bm{x})^2 \pi(\bm{x}) \mathrm{d}\bm{x} \leq \int_{\Omega} \|\psi\|_{\mathcal{H}_0}^2 k_0(\bm{x},\bm{x}) \pi(\bm{x}) \mathrm{d}\bm{x} = \|\psi\|_{\mathcal{H}_0}^2 \int_{\Omega}  k_0(\bm{x},\bm{x}) \pi(\bm{x}) \mathrm{d}\bm{x} < \infty,
\end{eqnarray*}
as required.
\end{proof}

\begin{proof}[Proof of Theorem \ref{sun proof}]
From Theorem \ref{conjugate}, (A1,3) ensure $\mathcal{H}_+$ is well-defined.
Unbiasedness follows from (A1,2',3) and Lemma \ref{finite variance}.
Below we employ the standard notation
$$
L^2(\pi) = \mathcal{L}^2(\pi) \setminus \{f \text{ such that } f = 0 \; \pi\text{-almost everywhere}\}
$$
and denote the standard norm on this space by $\|\cdot\|_{L^2(\pi)}$.

For the remainder we appeal to the relatively recent work of \cite{Sun1}, who considered convergence in a general setting where (i) $\Omega \cup \partial\Omega$ is not required to be compact in $\mathbb{R}^d$, and (ii) only weak assumptions are required on the kernel $k_+$, which can be easily satisfied in our setting.
To this end, define the integral operator 
$$
(T g)(\bm{x}) := \int_{\Omega}  k_+(\bm{x},\bm{x}') g(\bm{x}') \pi(\bm{x}') \mathrm{d}\bm{x}', \; \; \; \; \; \bm{x} \in \Omega, \; g \in L^2(\pi).
$$
In the well-posed setting of (A5), Theorem 1.1 of \cite{Sun1} establishes that if (i) $\sup_{\bm{x} \in \Omega} k_+(\bm{x},\bm{x}) < \infty$ and (ii) $T^{-1/2} f \in L^2(\pi)$, then with a RLS estimator based on $\lambda = O(m^{-1/2})$ we have $\mathbb{E}_{\mathcal{D}_0} [ \sigma^2(f - s_{f,\mathcal{D}_0})] = O(m^{-1/6})$.
Inserting this rate into Proposition \ref{theo1} with $\gamma = 1$, $\delta = 1/6$ would produce a MSE $\mathbb{E}_{\mathcal{D}_0} \mathbb{E}_{\mathcal{D}_1}[(\hat{\mu}(\mathcal{D}_0,\mathcal{D}_1;f) - \mu(f))^2] = O(n^{-1-\gamma\delta}) = O(n^{-7/6})$.

It therefore remains to prove requirements (i) and (ii) above are satisfied.
For (i) we have that $\sup_{\bm{x} \in \Omega} k_+(\bm{x},\bm{x}') = 1 + \sup_{\bm{x} \in \Omega} k_0(\bm{x},\bm{x})$, where the second term is finite by (A4').
For (ii), Prop. 3.3 of \cite{Sun1} (which does not depend on Theorem 1.1 of the same paper) shows that, when (i) holds, we have $\|T^{-1/2} h\|_{L^2(\pi)} = \|h\|_{\mathcal{H}_+}$ for all $h \in \mathcal{H}_+$.
Since $f \in \mathcal{H}_+$ by (A5) we thus have $\|T^{-1/2} f\|_{L^2(\pi)} = \|f\|_{\mathcal{H}_+} < \infty$ and so $T^{-1/2}f \in L^2(\pi)$, as required.
\end{proof}

\begin{proof}[Proof of Lemma \ref{explicit formulae}]
From Theorem \ref{conjugate}, (A1,3) ensure $\mathcal{H}_0$ is well-defined.
The interpolation problem is equivalently expressed as $s_{f,\mathcal{D}_0} = \hat{c} + \hat{\psi}$ where
\begin{eqnarray*}
(\hat{c},\hat{\psi}) := \underset{c \in \mathcal{C}, \; \psi \in \mathcal{H}_0}{\arg\min} \|c\|_{\mathcal{C}}^2 + \|\psi\|_{\mathcal{H}_0}^2 \text{ s.t. } f(\bm{x}_j) = c + \psi(\bm{x}_j) \text{ for } j = 1,...,m.
\end{eqnarray*}
For fixed $c \in \mathcal{C}$, the representer theorem \citep[][Thm. 5.5, p168]{Steinwart2} tells us that the solution
\begin{eqnarray*}
\hat{\psi} = \underset{\psi \in \mathcal{H}_0}{\arg\min} \|\psi\|_{\mathcal{H}_0} \text{ s.t. } f(\bm{x}_j) = c + \psi(\bm{x}_j) \text{ for } j = 1,...,m 
\end{eqnarray*}
takes the form $\psi(\bm{x}) = \sum_{i=1}^m \beta_i k_0(\bm{x}_i,\bm{x})$ where, due to the reproducing property, $\|\psi\|_{\mathcal{H}_0}^2 = \bm{\beta}^T\bm{K}_0\bm{\beta}$.
Thus writing $\bm{\beta} = [\beta_1,\dots,\beta_m]^T$ reduces the problem to
\begin{eqnarray*}
(\hat{c},\hat{\bm{\beta}}) = \underset{c \in \mathcal{C} , \bm{\beta} \in \mathbb{R}^m}{\arg\min} c^2 \text{ s.t. } f(\bm{x}_j) = c + \sum_{i=1}^m \beta_i k_0(\bm{x}_i,\bm{x}_j) \text{ for } j = 1,...,m
\end{eqnarray*}
Differentiating with respect to $c$ and $\bm{\beta}$ leads, via the Woodbury matrix inversion identity, to the solution
\begin{eqnarray*}
\hat{c} = \frac{\bm{1}^T\bm{K}_0^{-1}\bm{f}_0}{1 + \bm{1}^T\bm{K}_0^{-1}\bm{1}}, \; \; \; \; \; \hat{\bm{\beta}} = \bm{K}_0^{-1} (\bm{f}_0 - \hat{c}\bm{1})
\end{eqnarray*}
and associated fitted values $\hat{\bm{f}}_1 = \hat{c}\bm{1} + \bm{K}_{1,0}\hat{\bm{\beta}}$ at the points $\mathcal{D}_1$.
Putting this together, we have
\begin{eqnarray*}
\hat{\mu}(\mathcal{D}_0,\mathcal{D}_1;f) \; = \; \frac{1}{n-m} \sum_{i=m+1}^n f_{\mathcal{D}_0}(\bm{x}_i) & = & \frac{1}{n-m} \sum_{i=m+1}^n f(\bm{x}_i) - s_{f,\mathcal{D}_0}(\bm{x}_i) + \mu(s_{f,\mathcal{D}_0}) \\
& = & \frac{1}{n-m} \bm{1}^T (\bm{f}_1 - \hat{\bm{f}}_1) + \hat{c}.
\end{eqnarray*}
This completes the proof.
\end{proof}

\begin{proof}[Proof of Theorem \ref{preasy}]
From Theorem \ref{conjugate}, (A1,3) ensure $\mathcal{H}_0$ is well-defined.
The CF estimator takes the form
$$
\hat{\mu}(\mathcal{D}_0,\mathcal{D}_1;f) = \sum_{i=1}^n w_i f(\bm{x}_i) = \hat{c} + \sum_{i=1}^n w_i \psi(\bm{x}_i)
$$
where, by Lemma \ref{explicit formulae}, the vector of weights $\bm{w} = [w_1,\dots,w_n]^T$ is given by
\begin{eqnarray}
\bm{w} = \left[ \begin{array}{c} - \frac{(\bm{K}_0 + \lambda m \bm{I})^{-1} \bm{K}_{0,1} \bm{1}}{n-m} + \frac{1}{n-m} \frac{(\bm{1}^T(\bm{K}_0  + \lambda m \bm{I})^{-1}\bm{K}_{0,1}\bm{1})  (\bm{K}_0  + \lambda m \bm{I})^{-1} \bm{1}}{1 + \bm{1}^T(\bm{K}_0 + \lambda m \bm{I})^{-1}\bm{1}} \\ \frac{1}{n-m} \bm{1} \end{array} \right] \label{weights vector}
\end{eqnarray}
and satisfies $\bm{1}^T\bm{w} = 1$.
Using the reproducing property, the estimation error is
\begin{eqnarray*}
\hat{\mu}(\mathcal{D}_0,\mathcal{D}_1;f) - \mu(f) & = & \sum_{i=1}^n w_i f(\bm{x}_i) - \int_{\Omega} f(\bm{x}) \pi(\bm{x}) \mathrm{d}\bm{x} \\
& = & \sum_{i=1}^n w_i \psi(\bm{x}_i) - \int_{\Omega} \psi(\bm{x}) \pi(\bm{x}) \mathrm{d}\bm{x} \\
& = & \Big\langle \psi , \sum_{i=1}^n w_i k_0(\cdot,\bm{x}_i) - \underbrace{\int_{\Omega} k_0(\cdot,\bm{x}) \pi(\bm{x}) \mathrm{d}\bm{x}}_{=0\text{ from Lemma \ref{mean ele lem}}} \Big\rangle_{\mathcal{H}_0}. \label{techn} 
\end{eqnarray*}
It follows from the Cauchy-Schwarz inequality that
$$
|\hat{\mu}(\mathcal{D}_0,\mathcal{D}_1;f) - \mu(f)| \leq \|\psi\|_{\mathcal{H}_0} \left\| \sum_{i=1}^n w_i k_0(\cdot,\bm{x}_i) \right\|_{\mathcal{H}_0}.
$$
The first term satisfies $\|\psi\|_{\mathcal{H}_0}^2 \leq \hat{c}^2 + \|\psi\|_{\mathcal{H}_0}^2 = \|f\|_{\mathcal{H}_+}^2$ and, from the reproducing property, the second term satisfies
\begin{eqnarray}
\left\| \sum_{i=1}^n w_i k_0(\cdot,\bm{x}_i) \right\|_{\mathcal{H}_0}^2 = \bm{w}^T \bm{K} \bm{w}, \; \; \; \; \; \bm{K} = \left[ \begin{array}{cc} \bm{K}_0 & \bm{K}_{0,1} \\ \bm{K}_{1,0} & \bm{K}_1 \end{array} \right]. \label{final form}
\end{eqnarray}
Finally, upon substituting Eqn. \ref{weights vector} into Eqn. \ref{final form} we obtain the required result with $D(\mathcal{D}_0,\mathcal{D}_1) = \bm{w}^T\bm{K}\bm{w}$.
The special case $\lambda = 0$ is reported in the statement of the theorem.
\end{proof}

\end{document}